\documentclass[11pt]{article}
\usepackage{amsmath,amsfonts,amssymb,amsthm}
\usepackage{graphicx}
\usepackage{hyperref}
\usepackage[margin=1in]{geometry}

\newtheorem{theorem}{Theorem}
\newtheorem{proposition}[theorem]{Proposition}

\newtheorem{lemma}[theorem]{Lemma}
\newtheorem{definition}[theorem]{Definition}
\newtheorem{remark}[theorem]{Remark}

\DeclareMathOperator{\Tr}{Tr}

\newcommand{\R}{\mathbb{R}}
\newcommand{\C}{\mathbb{C}}
\newcommand{\Z}{\mathbb{Z}}
\newcommand{\HH}{\mathcal{H}}
\newcommand{\FF}{\mathcal{F}}
\newcommand{\RR}{\mathcal{R}}

\title{Continuous Algebraic Diversity: Unifying Spectral, Wavelet, and\\
Time-Frequency Analysis via Lie Group Actions}
\author{Mitchell A. Thornton\\
Richardson, TX 75080 USA\\
\texttt{mitchat@sbcglobal.net}}
\date{April 2026}

\begin{document}
\maketitle

\begin{abstract}
We provide a computable criterion for selecting among the classical spectral analysis methods (Fourier, wavelet, time-frequency) by extending the algebraic diversity (AD) framework to Lie groups acting on continuous waveforms in $L^2(\R)$.  To our knowledge, there is no other criterion that provides this selection capability. The central construction, the group-averaged estimator, generalizes from a finite sum over group elements to an integral over the group with respect to its Haar measure.  We prove a Continuous Replacement Theorem establishing that the continuous group-averaged estimator separates signal and noise subspaces under equivariance and ergodicity conditions analogous to the discrete case.  The main result is a Unification Theorem showing that three major branches of signal analysis are special cases of continuous AD with different Lie groups: classical spectral analysis (Wiener-Khinchin) corresponds to the translation group, wavelet analysis corresponds to the affine group, and time-frequency analysis (ambiguity function, Wigner-Ville distribution) corresponds to the Heisenberg-Weyl group.  The commutativity residual extends to the Hilbert-Schmidt norm on operator commutators, providing a continuous matched-group criterion that classifies stochastic processes by their algebraic symmetry.  A Discretization Recovery Theorem shows that all discrete AD results are sampling approximations to the continuous theory, with the DFT, DCT, and KLT arising as restrictions of continuous Lie group actions to finite sample grids.
\end{abstract}

\section{Introduction}

The algebraic diversity (AD) framework~\cite{thornton2026ad_arxiv} establishes that temporal averaging over multiple observations can be generalized to algebraic group action on a single observation for second-order statistical estimation.  The General Replacement Theorem identifies conditions under which a \emph{group-averaged estimator}
\begin{equation}\label{eq:dad}
\mathbf{F}_G(\mathbf{x}) = \frac{1}{|G|}\sum_{g \in G} [\rho(g)\mathbf{x}][\rho(g)\mathbf{x}]^H
\end{equation}
constructed from a single observation $\mathbf{x} \in \C^M$ recovers the signal-noise subspace decomposition.  The theory has been developed entirely in the discrete setting: finite groups, finite-dimensional vectors, finite sums. While the discrete AD framework is powerful in our view, it is incomplete.  It applies to sampled data on uniform grids, and the groups it employs are finite.  Many problems of practical importance fall outside its scope.  Self-similar signals such as $1/f$ noise and fractional Brownian motion have power-law spectra that no finite group can represent exactly.  Chirp waveforms sweep continuously through frequency, and discretizing them into pulses introduces boundary artifacts at every pulse edge.  Spherical data, from cosmic microwave background maps to gravitational field models, lives on $S^2$ and requires the infinite rotation group $\mathrm{SO}(3)$.  As a consequence of the continuous AD framework, we show that the Heisenberg uncertainty principle, which governs the resolution limits of every spectral method, is a result grounded in the fact that two Lie groups cannot simultaneously be optimally matched to the same signal. This paper closes the gap between finite groups and Lie groups, bringing the following classes of problems within reach:

\begin{itemize}
\item \textbf{Principled transform selection.}  Engineers routinely choose between Fourier, wavelet, and time-frequency analysis by convention or trial and error.  The continuous commutativity residual $\delta$ provides, for the first time, a computable criterion that selects the optimal transform for a given signal from first principles.

\item \textbf{Self-similar and fractal signal analysis.}  The affine group (translations plus dilations) is the natural matched group for processes with power-law spectra.  The continuous framework makes this precise and provides the admissibility conditions under which wavelet analysis is optimal.

\item \textbf{Waveform diversity without pulse boundaries.}  Continuous group averaging operates directly on waveforms, eliminating the arbitrary pulse segmentation that discrete eigentensors require and providing access to sub-pulse structure.

\item \textbf{Spherical harmonic analysis as AD.}  The angular power spectrum $C_\ell$ used in CMB cosmology, geodesy, and molecular orbital theory is the group-averaged estimator for $\mathrm{SO}(3)$, and the AD framework provides the same matched-group optimality analysis that was previously available only for finite groups.

\item \textbf{Resolution limits from group theory.}  The uncertainty principle $\Delta\omega\,\Delta t \geq 1/2$ emerges as a group commutation constraint: the translation and modulation generators do not commute, so no signal can be simultaneously optimal for both groups.
\end{itemize}

In the continuous AD framework, the observation is now a waveform $x \in L^2(\R)$, the group is a Lie group $G$ with Haar measure $\mu$, and the group-averaged estimator in the discrete AD framework becomes an integral:
\begin{equation}\label{eq:cad}
\FF_G(x) = \int_G [\rho(g)x] \otimes [\rho(g)x]^* \, d\mu(g),
\end{equation}
where $\rho: G \to U(\HH)$ is a unitary representation on the Hilbert space $\HH = L^2(\R)$ and $\otimes$ denotes the tensor product.

The motivation for developing the continuous AD framework is twofold.  First, many signals of engineering interest, including radar waveforms, acoustic emissions, and biomedical recordings, are continuous in nature, and discretization into pulses introduces arbitrary boundary effects.  Working directly in the continuum eliminates these artifacts and provides access to structure at scales finer than the sampling interval.  Second, the continuous formulation reveals a deep unification: the three major branches of signal analysis (spectral, wavelet, and time-frequency) are all instances of continuous AD with different Lie groups (Figure~\ref{fig:discrete-continuous}).

\begin{figure}[ht]
\centering
\includegraphics[width=\textwidth]{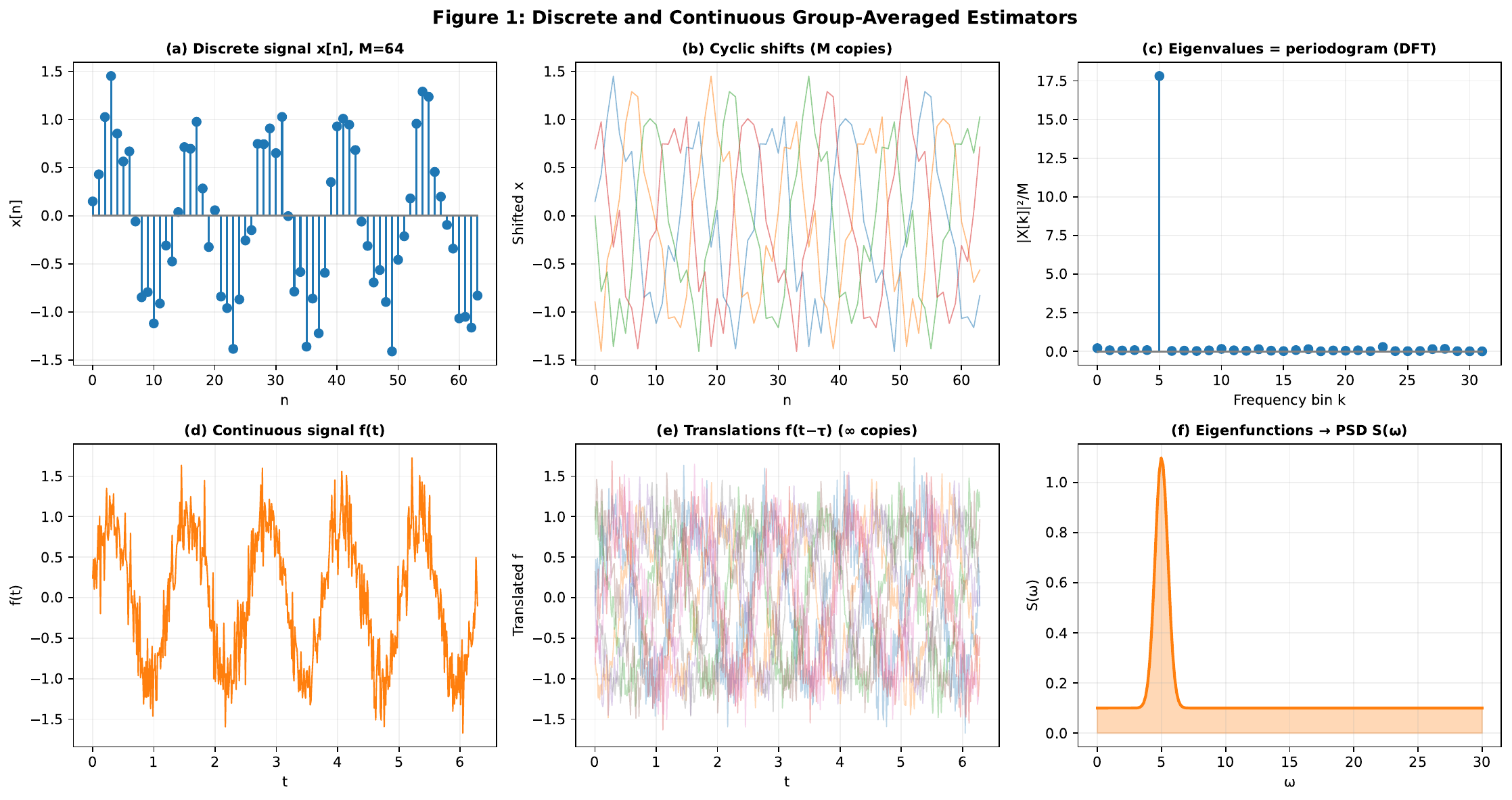}
\caption{The discrete-to-continuous passage.  Top: discrete signal, cyclic shifts, periodogram (DFT eigenvalues).  Bottom: continuous function, translations, power spectral density.  The algebraic construction is identical; only the group changes from $\Z_M$ to $(\R,+)$.}
\label{fig:discrete-continuous}
\end{figure}

\subsection{Principal Contributions}

The individual group-theoretic interpretations of Fourier, wavelet, and time-frequency analysis are known.  Grossmann and Morlet identified wavelets with the affine group in 1984.  Folland connected the STFT to the Heisenberg group in 1989.  Ali, Antoine, and Gazeau unified these as coherent state systems in 2000.  This paper does not rediscover those connections.  It contributes four results that the prior literature does not contain:

\begin{enumerate}
\item \textbf{A computable transform selection criterion.}  The commutativity residual $\delta$, extended from finite groups to operator commutators on $L^2(\R)$, provides the first principled method for choosing among Fourier, wavelet, and time-frequency analysis for a given signal.  Prior work tells you what each transform is.  The commutativity residual tells you which one to use.

\item \textbf{An explanation of the wavelet noise floor.}  The frequency-dependent noise floor in wavelet analysis is well known empirically; practitioners compensate for it with ad hoc normalization.  The Continuous Replacement Theorem traces it to a single cause: the Duflo-Moore operator $C_\rho$, which is nontrivial because the affine group is non-unimodular.  The translation group and the Heisenberg-Weyl group are unimodular, so their noise floors are flat.  This explains a long-standing asymmetry from first principles.

\item \textbf{A polynomial-time algorithm for blind group matching.}  The double-commutator generalized eigenvalue problem (GEVP) identifies the optimal group from a single observation in $O(M^3)$ time.  No prior work in harmonic analysis, algebraic signal processing, or coherent state theory provides an algorithm for selecting the optimal group from data.

\item \textbf{A discretization bridge.}  The Discretization Recovery Theorem establishes that the discrete AD framework is the sampled restriction of the continuous one, with $O(1/M)$ convergence.  Without this bridge, the discrete and continuous frameworks are separate results.  With it, they are one theory.
\end{enumerate}

These four principal contributions are justified by the following theorems and the continuous form of the commutativity residual that provide mathematical foundation.

\begin{enumerate}
\item \textbf{Continuous Replacement Theorem} (Theorem~\ref{thm:continuous_replacement}): Under signal equivariance and noise ergodicity conditions, $\FF_G(x)$ is a positive trace-class operator whose spectral decomposition separates signal and noise subspaces.

\item \textbf{Unification Theorem} (Theorem~\ref{thm:unification}): Classical spectral analysis is continuous AD with the translation group; wavelet analysis is continuous AD with the affine group; time-frequency analysis is continuous AD with the Heisenberg-Weyl group.

\item \textbf{Continuous Commutativity Residual} (Definition~\ref{def:continuous_delta}): The matched-group criterion extends to the Hilbert-Schmidt norm of the operator commutator, classifying stochastic processes by their algebraic symmetry.

\item \textbf{Matched Group Classification} (Theorem~\ref{thm:classification}): Stationary processes are matched to the translation group, self-similar processes to the affine group, and chirp processes to the Heisenberg-Weyl group.

\item \textbf{Discretization Recovery} (Theorem~\ref{thm:discretization}): Discrete AD converges to continuous AD as $M \to \infty$, with $\Z_M \to (\R,+)$, conjugated cyclic $\to$ affine, and $\Z_M \times \Z_M \to$ Heisenberg-Weyl.
\end{enumerate}

\section{Relation to Prior Work}\label{sec:prior}

Several bodies of work touch the territory of this paper.  We distinguish the present contribution from each.

\textbf{Wavelet theory and the affine group.}  Grossmann and Morlet~\cite{grossmann1984} established that the continuous wavelet transform is a square-integrable representation of the affine group, and Daubechies~\cite{daubechies1992} formalized the admissibility conditions and fast algorithms.  The Calder\'{o}n reproducing formula is the resolution of identity for the affine group representation.  These results are contained within the continuous AD framework as the special case $G = \R^+ \ltimes \R$ of the Unification Theorem.  What the wavelet literature does not provide is a criterion for deciding \emph{when} wavelets are the right tool.  The commutativity residual $\delta$ fills this gap: wavelets are optimal when $\delta(\mathrm{Aff}, \RR) < \delta((\R,+), \RR)$, i.e., when the signal's covariance is better matched to scale invariance than to translation invariance.

\textbf{Time-frequency analysis and the Heisenberg group.}  Folland~\cite{folland1989} and Gr\"{o}chenig~\cite{groechenig2001} developed the theory of the short-time Fourier transform, Gabor frames, and the Wigner-Ville distribution from the Heisenberg group perspective.  The Moyal identity, which we use in Proposition~\ref{prop:hw_dm}, is standard in this literature.  Again, the AD contribution is not the individual transform but the \emph{selection criterion}: given a signal, should one use Fourier, wavelet, or time-frequency analysis?  The prior literature treats each as a separate theory with its own optimality conditions.  The AD framework subsumes all three under a single construction and provides a uniform metric ($\delta$) for choosing among them.

\textbf{Algebraic signal processing (ASP).}  P\"{u}schel and Moura~\cite{puschel2008foundation,puschel2008space} developed ASP as an axiomatic framework in which a ``signal model'' (an algebra, a module, and a homomorphism) determines a class of linear transforms.  ASP provides the \emph{atlas}: given a signal model, it identifies the associated transform.  AD provides the \emph{compass}: given observed data, it identifies which signal model (and hence which transform) is optimal.  The two frameworks are complementary.  ASP tells you what the transforms are; AD tells you which one to use.  The continuous extension in this paper further separates the two: ASP operates on polynomial algebras and does not extend naturally to $L^2(\R)$, while AD extends directly via Haar integration.

\textbf{Coherent states.}  Ali, Antoine, and Gazeau~\cite{ali2000} provide a comprehensive treatment of coherent states, wavelets, and their generalizations from the viewpoint of square-integrable group representations.  The Duflo-Moore theorem, which we use as the convergence foundation in Section~\ref{sec:duflo_moore}, is central to their development.  The continuous group-averaged estimator can be viewed as the reproducing kernel of a coherent state system.  The novelty of the AD perspective is the estimation-theoretic interpretation: the group-averaged estimator is not merely a resolution of identity but a \emph{single-observation statistical estimator} with signal-noise separation properties (Theorem~\ref{thm:continuous_replacement}).

\textbf{Compressed sensing.}  Cand\`{e}s, Romberg, and Tao~\cite{candes2006} showed that signals sparse in some basis can be recovered from sub-Nyquist measurements.  AD and compressed sensing exploit orthogonal structural properties: sparsity $s$ (the number of nonzero coefficients) and group order $|G|$ (the number of symmetry-related copies).  A signal can be sparse but non-symmetric, symmetric but non-sparse, or both.  The two frameworks compose: AD can provide the basis in which sparsity is assessed, and compressed sensing can reduce the measurement burden for signals that are sparse in the AD-optimal basis.

\textbf{What is new here.}  The prior works cited above each develop the group-theoretic interpretation of a single transform family.  This paper does three things none of them does: (1)~it unifies all transform families under a single group-averaged estimator with uniform signal-equivariance and noise-ergodicity conditions; (2)~it provides the commutativity residual $\delta$ as a computable criterion for selecting among Lie groups; and (3)~it proves the discretization recovery theorem connecting the continuous and discrete frameworks.

\section{Mathematical Preliminaries}

\subsection{Locally Compact Groups and Haar Measure}

Let $G$ be a locally compact topological group.  A \emph{left Haar measure} on $G$ is a nonzero Radon measure $\mu$ satisfying $\mu(gE) = \mu(E)$ for all $g \in G$ and all Borel sets $E \subseteq G$.  Such a measure exists and is unique up to a positive multiplicative constant.  For compact groups, we normalize $\mu(G) = 1$.  For the non-compact groups of interest, the left Haar measures are:
\begin{itemize}
\item Translation group $(\R, +)$: $d\mu(\tau) = d\tau$ (Lebesgue measure).  Unimodular.
\item Affine group $(\R^+ \ltimes \R)$: $d\mu(a,b) = \frac{da\,db}{a^2}$.  \textbf{Non-unimodular:} the modular function is $\Delta(a,b) = 1/a$.
\item Heisenberg-Weyl group $H$: $d\mu(\tau,\nu,\phi) = d\tau\,d\nu\,d\phi$.  Unimodular.
\end{itemize}
The non-unimodularity of the affine group has direct consequences for the noise structure of the group-averaged estimator (see Remark~\ref{rem:noise_floor}).

\subsection{Unitary Representations and Square-Integrability}

A \emph{unitary representation} of $G$ on a Hilbert space $\HH$ is a group homomorphism $\rho: G \to U(\HH)$ that is strongly continuous.

\begin{definition}[Square-Integrable Representation]\label{def:sq_int}
A unitary representation $(\rho, \HH)$ of $G$ is \emph{square-integrable} (or \emph{discrete series}) if there exists a nonzero vector $\psi \in \HH$ (called an \emph{admissible vector}) such that
\begin{equation}
\int_G |\langle \rho(g)\psi, \psi \rangle|^2 \, d\mu(g) < \infty.
\end{equation}
For groups with nontrivial center $Z$, we say $\rho$ is \emph{square-integrable modulo $Z$} if the integral converges over the quotient $G/Z$.
\end{definition}

The three groups of interest have fundamentally different square-integrability properties, which we now characterize.

\subsection{Trace-Class and Hilbert-Schmidt Operators}

An operator $A$ on $\HH$ is \emph{trace-class} if $\sum_k \langle |A| e_k, e_k \rangle < \infty$ for some (and hence every) orthonormal basis $\{e_k\}$.  It is \emph{Hilbert-Schmidt} if $\|A\|_{HS}^2 := \sum_k \|A e_k\|^2 < \infty$.  The Hilbert-Schmidt norm generalizes the Frobenius norm from matrices to operators.

\section{The Duflo-Moore Framework}\label{sec:duflo_moore}

The group-averaged estimator~\eqref{eq:cad} is an integral.  Integrals can diverge.  The question is: when does this one converge, and what controls the answer?

For finite groups, the question does not arise.  A finite sum of finite-rank operators is always finite.  For Lie groups, the integral is over a non-compact domain, and convergence depends on how fast the integrand oscillates.  If the representation $\rho$ is ``square-integrable,'' the oscillation is fast enough to produce cancellation, and the integral converges.  If it is not, we need a different argument.

The three groups of interest fall into three different categories.  The affine group has a square-integrable representation, and convergence follows from a classical result of Duflo and Moore.  The Heisenberg-Weyl group is square-integrable only after quotienting out its center.  The translation group is not square-integrable at all, and convergence relies on the signal being in $L^2$ rather than on any property of the representation.  These are not three versions of the same argument.  They are three genuinely different mechanisms, and the differences explain real phenomena: why wavelets need an admissibility condition, why the STFT does not, and why the Fourier transform works for any finite-energy signal.

\begin{theorem}[Duflo-Moore Orthogonality Relations~\cite{duflo1976}]\label{thm:duflo_moore}
Let $G$ be a locally compact group with left Haar measure $\mu$, and let $(\rho, \HH)$ be an irreducible square-integrable unitary representation.  Then there exists a unique positive, self-adjoint, densely defined operator $C_\rho$ (the \emph{Duflo-Moore operator}) such that for all $\psi_1, \psi_2 \in \mathrm{Dom}(C_\rho)$ and all $\phi_1, \phi_2 \in \HH$:
\begin{equation}\label{eq:dm_orthog}
\int_G \langle \rho(g)\psi_1, \phi_1\rangle \overline{\langle \rho(g)\psi_2, \phi_2\rangle} \, d\mu(g) = \langle C_\rho \psi_1, C_\rho \psi_2\rangle \langle \phi_2, \phi_1\rangle.
\end{equation}
If $G$ is unimodular, then $C_\rho = d_\rho^{-1/2} I$ where $d_\rho > 0$ is the \emph{formal degree} of $\rho$, and~\eqref{eq:dm_orthog} reduces to the Schur orthogonality relations.
\end{theorem}

We now compute $C_\rho$ explicitly for each of the three groups.

\begin{proposition}[Affine Group: Explicit Duflo-Moore Operator]\label{prop:affine_dm}
Let $G_{\mathrm{Aff}} = \R^+ \ltimes \R$ with left Haar measure $d\mu(a,b) = a^{-2}\,da\,db$ and representation $\rho(a,b)x(t) = a^{-1/2}x((t-b)/a)$ on $L^2(\R)$.  Then:
\begin{enumerate}
\item[(i)] The representation $\rho$ decomposes into two irreducible subrepresentations on $L^2_+(\R) = \{x : \hat{x}(\omega) = 0 \text{ for } \omega < 0\}$ and $L^2_-(\R) = \{x : \hat{x}(\omega) = 0 \text{ for } \omega > 0\}$.
\item[(ii)] On each irreducible component, the Duflo-Moore operator is multiplication by $|\omega|^{-1/2}$ in the frequency domain:
\begin{equation}
(C_\rho \psi)\,\widehat{\phantom{x}}(\omega) = |\omega|^{-1/2}\hat{\psi}(\omega).
\end{equation}
\item[(iii)] A vector $\psi \in L^2(\R)$ is admissible if and only if the \emph{Calder\'{o}n admissibility constant}
\begin{equation}\label{eq:calderon}
c_\psi := \int_0^{\infty} \frac{|\hat{\psi}(\omega)|^2}{\omega} \, d\omega < \infty.
\end{equation}
This holds for any $\psi$ with $\hat{\psi}(0) = 0$ (zero mean) and sufficient decay.
\end{enumerate}
\end{proposition}

\begin{proof}
Part~(i) follows from the fact that the Fourier transform intertwines $\rho(a,b)$ with multiplication by $e^{-ib\omega}$ and dilation $\omega \mapsto a\omega$; the half-lines $\omega > 0$ and $\omega < 0$ are invariant.

For~(ii), compute $(C_\rho \psi)\,\widehat{\phantom{x}}$ by applying Theorem~\ref{thm:duflo_moore} with $\psi_1 = \psi_2 = \phi_1 = \phi_2 = \psi$ and using the explicit form of the wavelet coefficient $\langle \rho(a,b)x, \psi\rangle = a^{-1/2}\int x(t)\overline{\psi((t-b)/a)}\,dt$.  In the frequency domain, this becomes $\sqrt{a}\int \hat{x}(\omega)\overline{\hat{\psi}(a\omega)}e^{ib\omega}d\omega$.  Integrating $|\langle \rho(a,b)x, \psi\rangle|^2$ over $b$ first (using Parseval) eliminates the cross-terms, and integrating over $a$ with measure $da/a^2$ yields the factor $\int |\hat{\psi}(a\omega)|^2 da/a = c_\psi/|\omega|$.  Comparing with~\eqref{eq:dm_orthog} identifies $\|C_\rho \psi\|^2 = c_\psi$, hence $C_\rho$ acts as multiplication by $|\omega|^{-1/2}$.

Part~(iii) follows immediately: $\psi \in \mathrm{Dom}(C_\rho)$ requires $C_\rho \psi \in L^2$, i.e., $\int |\omega|^{-1}|\hat{\psi}(\omega)|^2 d\omega < \infty$, which is~\eqref{eq:calderon}.
\end{proof}

\begin{proposition}[Heisenberg-Weyl Group: Square-Integrability Modulo Center]\label{prop:hw_dm}
Let $H = \{(\tau, \nu, \phi) : \tau, \nu, \phi \in \R\}$ be the Heisenberg-Weyl group with multiplication $(\tau_1,\nu_1,\phi_1)\cdot(\tau_2,\nu_2,\phi_2) = (\tau_1+\tau_2, \nu_1+\nu_2, \phi_1+\phi_2+\tau_1\nu_2)$ and the Schr\"{o}dinger representation $\rho(\tau,\nu,\phi)x(t) = e^{i\phi}e^{i2\pi\nu t}x(t-\tau)$ on $L^2(\R)$.  Then:
\begin{enumerate}
\item[(i)] By the Stone-von Neumann theorem, $\rho$ is (up to unitary equivalence) the unique irreducible representation of $H$ with central character $e^{i\phi}$.
\item[(ii)] The representation is \textbf{not} square-integrable over $H$, because the center $Z = \{(0,0,\phi) : \phi \in \R\}$ acts by scalars.
\item[(iii)] However, $\rho$ is square-integrable modulo $Z$.  On the quotient $H/Z \cong \R^2$ (the time-frequency plane), \textbf{every} nonzero $\psi \in L^2(\R)$ is admissible, and the Moyal identity gives
\begin{equation}\label{eq:moyal}
\int_{\R^2} |\langle x, \rho(\tau,\nu,0)\psi\rangle|^2 \, d\tau\,d\nu = \|x\|^2 \|\psi\|^2.
\end{equation}
\item[(iv)] The Duflo-Moore operator on $H/Z$ is $C_\rho = I$ (the identity), reflecting the unimodularity of $H/Z \cong \R^2$.
\end{enumerate}
\end{proposition}

\begin{proof}
Part~(i) is the Stone-von Neumann theorem~\cite{folland1989}.

For~(ii), $|\langle \rho(\tau,\nu,\phi)x, \psi\rangle| = |\langle \rho(\tau,\nu,0)x, \psi\rangle|$ is independent of $\phi$, so the integral over $\phi \in \R$ diverges.

Part~(iii): The expression $V_\psi x(\tau,\nu) := \langle x, \rho(\tau,\nu,0)\psi\rangle = \int x(t)\overline{\psi(t-\tau)}e^{-i2\pi\nu t}\,dt$ is the short-time Fourier transform (STFT).  The Moyal identity~\eqref{eq:moyal} is a standard result in time-frequency analysis~\cite{folland1989,groechenig2001}: it follows from applying Parseval's theorem in $\nu$ (which collapses the frequency integral) and then in $\tau$ (which collapses the time integral).  No admissibility condition is needed beyond $\psi \neq 0$.

Part~(iv) follows from comparing~\eqref{eq:moyal} with~\eqref{eq:dm_orthog}: the right-hand side of~\eqref{eq:dm_orthog} with $C_\rho = I$ gives $\|\psi\|^2\|x\|^2 = \|x\|^2\|\psi\|^2$, consistent.
\end{proof}

\begin{proposition}[Translation Group: Convergence Without Square-Integrability]\label{prop:translation}
The translation group $(\R, +)$ is Abelian, unimodular, and non-compact.  Its irreducible representations are the one-dimensional characters $\chi_\omega(\tau) = e^{i\omega\tau}$, which are \textbf{not} square-integrable: $\int_\R |e^{i\omega\tau}|^2 d\tau = \infty$.  The Duflo-Moore theorem does not apply.  Nevertheless, the group-averaged estimator converges:
\begin{enumerate}
\item[(i)] For any $x \in L^2(\R)$, the kernel
\begin{equation}\label{eq:autocorr}
[\FF_{(\R,+)}(x)](t,s) = \int_{-\infty}^{\infty} x(t+\tau)\overline{x(s+\tau)}\,d\tau =: R_{xx}(t-s)
\end{equation}
converges absolutely for every $(t,s)$, with $|R_{xx}(t-s)| \leq \|x\|^2$.
\item[(ii)] By the Plancherel theorem, $\FF_{(\R,+)}(x)$ is a convolution operator unitarily equivalent to multiplication by the energy spectral density $|\hat{x}(\omega)|^2$ in the frequency domain.
\item[(iii)] The trace satisfies $\Tr(\FF_{(\R,+)}(x)) = R_{xx}(0) = \|x\|^2$.
\end{enumerate}
\end{proposition}

\begin{proof}
Part~(i) follows from Cauchy-Schwarz:
$|R_{xx}(t-s)| = \left|\int x(t+\tau)\overline{x(s+\tau)}\,d\tau\right| \leq \|x(\cdot+t)\|_{L^2}\|x(\cdot+s)\|_{L^2} = \|x\|^2$
by translation invariance of the $L^2$ norm.

Part~(ii): Taking the Fourier transform of $R_{xx}(\Delta) = \int x(t)\overline{x(t-\Delta)}\,dt$ with respect to $\Delta$, the correlation theorem (a consequence of Plancherel) gives $\widehat{R_{xx}}(\omega) = |\hat{x}(\omega)|^2$.

Part~(iii): $R_{xx}(0) = \int |x(t)|^2 dt = \|x\|^2$.
\end{proof}

\begin{remark}\label{rem:convergence_types}
The three groups exhibit three distinct convergence mechanisms for the group-averaged estimator:
\begin{itemize}
\item \textbf{Affine group:} Convergence via the Duflo-Moore theorem; requires an admissible wavelet $\psi$ satisfying the Calder\'{o}n condition~\eqref{eq:calderon}.
\item \textbf{Heisenberg-Weyl group (mod center):} Convergence via the Moyal identity; \emph{every} nonzero window $\psi$ is admissible.
\item \textbf{Translation group:} Convergence via Cauchy-Schwarz on $L^2$; the representation is not square-integrable, but the integrand has finite $L^1$ norm because $x \in L^2$.
\end{itemize}
These differences are not pathological; they reflect genuine physical distinctions among the three analysis methods and explain, for instance, why wavelet analysis requires a carefully chosen wavelet while the Fourier transform does not.
\end{remark}

\section{The Continuous Group-Averaged Estimator}\label{sec:estimator}

With the convergence machinery in place, we can state the central definition.  The idea is simple.  Take a waveform $x$.  Apply every group action $\rho(g)$ to it.  At each $g$, form the rank-one outer product $[\rho(g)x]\otimes[\rho(g)x]^*$.  Integrate over the group.  The result is a positive operator whose eigenfunctions are the ``natural basis'' for $x$ relative to the group $G$, and whose eigenvalues measure the energy in each basis component.

\begin{definition}[Continuous Group-Averaged Estimator]\label{def:cad_estimator}
Let $x \in L^2(\R)$, let $G$ be a locally compact group with Haar measure $\mu$, and let $\rho: G \to U(L^2(\R))$ be a unitary representation.  The \emph{continuous group-averaged estimator} is the operator
\begin{equation}\label{eq:cad_estimator}
\FF_G(x) = \int_G [\rho(g)x] \otimes [\rho(g)x]^* \, d\mu(g),
\end{equation}
where $[f \otimes h^*](\cdot) = \langle \cdot, h \rangle f$ is the rank-one operator.  The integral is interpreted according to the convergence mechanism appropriate to $G$ (Remark~\ref{rem:convergence_types}).
\end{definition}

\begin{lemma}[Integrand Norm Invariance]\label{lem:integrand_norm}
For any unitary representation $\rho$ and any $x \in L^2(\R)$:
\begin{equation}
\|\rho(g)x \otimes (\rho(g)x)^*\|_{HS} = \|x\|^2 \quad \text{for all } g \in G.
\end{equation}
\end{lemma}
\begin{proof}
$\|\rho(g)x \otimes (\rho(g)x)^*\|_{HS}^2 = \|\rho(g)x\|^4 = \|x\|^4$ by unitarity.
\end{proof}

This lemma shows that every integrand in~\eqref{eq:cad_estimator} has the same Hilbert-Schmidt norm.  Convergence therefore depends entirely on the oscillatory cancellation of the integrand, which is governed by the square-integrability properties established in Section~\ref{sec:duflo_moore}.

\section{The Continuous Replacement Theorem}\label{sec:replacement}

The discrete Replacement Theorem says: if the signal is equivariant under the group and the noise is ergodic, then the group-averaged estimator separates them.  One eigenvalue is large (the signal).  The rest are small (the noise).  The theorem works because the group action coherently reinforces the signal while incoherently averaging the noise.

The continuous version says the same thing, but with a twist.  For two of the three groups (translation and Heisenberg-Weyl), the noise floor is flat: white noise stays white after group averaging.  For the third (the affine group), the noise floor is frequency-dependent.  High frequencies are noisier than low frequencies.  This asymmetry is well known in wavelet analysis, but its cause has not been clearly identified.  The AD framework traces it to a single source: the modular function of the affine group, which enters through the Duflo-Moore operator $C_\rho$.

\begin{theorem}[Continuous Replacement Theorem]\label{thm:continuous_replacement}
Let $x = s + n \in L^2(\R)$ with $\|x\| < \infty$.  Let $G$ be one of the three groups: $(\R,+)$, $G_{\mathrm{Aff}}$, or $H/Z$.  Define $\FF_G(x)$ as in Definition~\ref{def:cad_estimator}, with convergence as established in Propositions~\ref{prop:affine_dm},~\ref{prop:hw_dm}, and~\ref{prop:translation}.  Suppose:
\begin{itemize}
\item[\textbf{(C1)}] \textbf{Signal equivariance:} $s$ lies in a finite-dimensional invariant subspace $\mathcal{S} \subset L^2(\R)$ under $\rho(G)$.
\item[\textbf{(C2)}] \textbf{Noise ergodicity:} $E\!\left[\int_G \rho(g)n \otimes (\rho(g)n)^* \, d\mu(g)\right] = \sigma^2 \mathcal{N}_G$,

where $\mathcal{N}_G$ is the \emph{noise operator} determined by the Duflo-Moore operator:
\begin{equation}\label{eq:noise_operator}
\mathcal{N}_G = \begin{cases}
\mathcal{I} & \text{if } G = (\R,+) \text{ or } H/Z \text{ (unimodular)}, \\
C_\rho^{-2} & \text{if } G = G_{\mathrm{Aff}} \text{ (non-unimodular)}.
\end{cases}
\end{equation}
\end{itemize}
Then:
\begin{enumerate}
\item[(i)] \textbf{Decomposition:} $E[\FF_G(x)] = \FF_G(s) + \sigma^2 \mathcal{N}_G$, where the cross-term $E[\mathcal{C}_{sn}] = 0$ by independence of $s$ and $n$.
\item[(ii)] \textbf{Signal concentration:} $\FF_G(s)$ has finite rank equal to $\dim(\mathcal{S})$.
\item[(iii)] \textbf{Noise structure:} For the unimodular groups, the noise contributes a flat floor ($\sigma^2 \mathcal{I}$).  For the affine group, the noise floor is frequency-dependent: $\sigma^2 |\omega|$ in the frequency domain, reflecting the action of $C_\rho^{-2}$.
\item[(iv)] \textbf{Subspace separation:} For SNR $\gg 1$, the eigenfunctions of $\FF_G(x)$ associated with the $\dim(\mathcal{S})$ largest eigenvalues converge to the signal subspace $\mathcal{S}$.
\end{enumerate}
\end{theorem}

\begin{proof}
\textit{Part~(i).}  Expanding $x = s + n$ in~\eqref{eq:cad_estimator} and using linearity yields $\FF_G(x) = \FF_G(s) + \FF_G(n) + \mathcal{C}_{sn}$, where $\mathcal{C}_{sn} = \int_G [\rho(g)s \otimes (\rho(g)n)^* + \rho(g)n \otimes (\rho(g)s)^*] d\mu(g)$.  Since $s$ is deterministic and $E[n] = 0$, $E[\mathcal{C}_{sn}] = 0$.

\textit{Part~(ii).}  By~(C1), $\rho(g)s \in \mathcal{S}$ for all $g \in G$, so $\rho(g)s \otimes (\rho(g)s)^*$ maps into $\mathcal{S}$ for every $g$.  The integral $\FF_G(s)$ therefore has range contained in $\mathcal{S}$, giving $\mathrm{rank}(\FF_G(s)) \leq \dim(\mathcal{S})$.

\textit{Part~(iii).}  For unimodular $G$ ($(\R,+)$ and $H/Z$): condition~(C2) with $\mathcal{N}_G = \mathcal{I}$ states that the group-averaged noise covariance is white.  This is the continuous analog of the discrete condition that the group representation decorrelates the noise equally in all directions.

For $G_{\mathrm{Aff}}$ (non-unimodular): the Duflo-Moore orthogonality relation~\eqref{eq:dm_orthog} introduces $C_\rho$ into the integral.  Specifically, for white noise $n$ with $E[n \otimes n^*] = \sigma^2 \mathcal{I}$, the Duflo-Moore relation with $\psi_1 = \psi_2 = n$ gives
\begin{equation}
E[\FF_{G_{\mathrm{Aff}}}(n)] = \sigma^2 C_\rho^{-2},
\end{equation}
where $C_\rho^{-2}$ acts as multiplication by $|\omega|$ in the frequency domain (since $C_\rho$ multiplies by $|\omega|^{-1/2}$).  This is the well-known frequency-dependent noise floor in wavelet analysis: high-frequency wavelet coefficients have higher noise variance than low-frequency coefficients.  The AD framework reveals this asymmetry as a direct consequence of the non-unimodularity of the affine group.

\textit{Part~(iv).}  From~(i)--(iii), $E[\FF_G(x)] = \FF_G(s) + \sigma^2 \mathcal{N}_G$.  The operator $\FF_G(s)$ has at most $\dim(\mathcal{S})$ nonzero eigenvalues, which dominate for large SNR.  Eigenfunctions associated with the $\dim(\mathcal{S})$ largest eigenvalues converge to $\mathcal{S}$ by operator perturbation theory (Weyl's inequality adapted to Hilbert-Schmidt operators).
\end{proof}

\begin{remark}[Frequency-Dependent Noise Floor]\label{rem:noise_floor}
The noise operator $\mathcal{N}_G = C_\rho^{-2}$ for the affine group explains a long-standing asymmetry in signal processing practice.  In Fourier analysis (translation group, unimodular), white noise remains white after transformation: the noise floor is flat across all frequencies.  In wavelet analysis (affine group, non-unimodular), white noise becomes colored: the noise power scales as $|\omega|$, giving higher noise at higher frequencies.  The AD framework identifies the root cause: the modular function $\Delta(a,b) = 1/a$ of the affine group, which enters through the Duflo-Moore operator.  No such asymmetry occurs for the Heisenberg-Weyl group (unimodular), where the STFT preserves whiteness.
\end{remark}

\section{The Unification Theorem}\label{sec:unification}

The Fourier transform, the wavelet transform, and the ambiguity function were invented decades apart, by different people, for different problems.  They are taught in different courses.  They appear in different textbooks.  They have different notation, different conventions, and different communities of practitioners.

They are the same construction.

Each one takes a signal, applies a family of group actions (translations, dilations, time-frequency shifts), forms the outer product at each group element, and integrates over the group.  The only difference is which group.  The following theorem makes this precise.

\begin{theorem}[Unification]\label{thm:unification}
The following classical transforms are instances of the continuous group-averaged estimator $\FF_G(x)$:

\begin{enumerate}
\item[(i)] \textbf{Spectral analysis.}  Let $G = (\R, +)$ with $\rho(\tau)x(t) = x(t + \tau)$.  Then
\begin{equation}
\FF_{(\R,+)}(x)(t,s) = \int_{-\infty}^{\infty} x(t+\tau)\overline{x(s+\tau)} \, d\tau = R_{xx}(t-s),
\end{equation}
which is the autocorrelation function.  Its spectral decomposition (via the Fourier transform) yields the Wiener-Khinchin power spectral density.

\item[(ii)] \textbf{Wavelet analysis.}  Let $G = \R^+ \ltimes \R$ (the affine group) with $\rho(a,b)x(t) = a^{-1/2}x((t-b)/a)$.  Then
\begin{equation}
\FF_{\mathrm{Aff}}(x) = \int_0^{\infty}\int_{-\infty}^{\infty} |\langle x, \psi_{a,b}\rangle|^2 \frac{da\,db}{a^2},
\end{equation}
where $\psi_{a,b}(t) = a^{-1/2}\psi((t-b)/a)$ is the wavelet family.  The diagonal of this operator in the wavelet basis is the scalogram.

\item[(iii)] \textbf{Time-frequency analysis.}  Let $G$ be the Heisenberg-Weyl group with $\rho(\tau,\nu)x(t) = e^{i2\pi\nu t}x(t - \tau)$.  Then
\begin{equation}
\FF_{\mathrm{HW}}(x)(\tau,\nu) = \int_{-\infty}^{\infty} x(t)\overline{x(t-\tau)}e^{-i2\pi\nu t} \, dt = A_x(\tau,\nu),
\end{equation}
which is the ambiguity function.  Its symplectic Fourier transform is the Wigner-Ville distribution.
\end{enumerate}
\end{theorem}

\begin{proof}
Each case is verified by substituting the group action and Haar measure into~\eqref{eq:cad_estimator} and applying the convergence results from Section~\ref{sec:duflo_moore}.

\textit{Part (i).}  With $\rho(\tau)x(t) = x(t+\tau)$ and $d\mu = d\tau$, the rank-one kernel is $[\rho(\tau)x \otimes (\rho(\tau)x)^*](t,s) = x(t+\tau)\overline{x(s+\tau)}$.  Integrating over $\tau$:
\begin{equation}
[\FF_{(\R,+)}(x)](t,s) = \int_{-\infty}^{\infty} x(t+\tau)\overline{x(s+\tau)}\,d\tau.
\end{equation}
The substitution $u = s + \tau$ shows this depends only on $t - s$, giving $R_{xx}(t-s)$.  Convergence is guaranteed by Proposition~\ref{prop:translation}(i).  The spectral decomposition follows from the Plancherel theorem: $R_{xx}$ is a convolution kernel whose Fourier transform is $|\hat{x}(\omega)|^2$.

\textit{Part (ii).}  With $\rho(a,b)x(t) = a^{-1/2}x((t-b)/a)$ and $d\mu = a^{-2}da\,db$, define the wavelet coefficient $W_\psi x(a,b) = \langle x, \psi_{a,b}\rangle$ where $\psi_{a,b}(t) = a^{-1/2}\psi((t-b)/a)$.  The diagonal of $\FF_{G_{\mathrm{Aff}}}(x)$ in the wavelet basis is $|W_\psi x(a,b)|^2$, and the Duflo-Moore relation (Proposition~\ref{prop:affine_dm}) gives
\begin{equation}
\int_0^\infty \int_{-\infty}^{\infty} |W_\psi x(a,b)|^2 \frac{da\,db}{a^2} = c_\psi \|x\|^2,
\end{equation}
where $c_\psi$ is the Calder\'{o}n constant~\eqref{eq:calderon}.  The resolution of identity $c_\psi^{-1}\int\int W_\psi x(a,b)\,\psi_{a,b}\,a^{-2}da\,db = x$ follows by polarization of the Duflo-Moore relation.

\textit{Part (iii).}  With $\rho(\tau,\nu,0)x(t) = e^{i2\pi\nu t}x(t-\tau)$ and $d\mu = d\tau\,d\nu$ on $H/Z$, the STFT coefficient is $V_\psi x(\tau,\nu) = \int x(t)\overline{\psi(t-\tau)}e^{-i2\pi\nu t}\,dt$.  Setting $\psi = x$ (the signal as its own window) gives the cross-ambiguity function:
\begin{equation}
A_x(\tau,\nu) = V_x x(\tau,\nu) = \int_{-\infty}^{\infty} x(t)\overline{x(t-\tau)}e^{-i2\pi\nu t}\,dt.
\end{equation}
The Moyal identity (Proposition~\ref{prop:hw_dm}(iii)) guarantees $\int\int |V_\psi x|^2 d\tau\,d\nu = \|x\|^2\|\psi\|^2 < \infty$, so the integral converges absolutely.  The Wigner-Ville distribution is obtained as the symplectic Fourier transform of $A_x$~\cite{folland1989}.
\end{proof}

\subsection{Rotation Groups: Circular and Spherical Harmonics}

The unification extends to the rotation groups, which govern signals on circles and spheres.

\textbf{SO(2) and circular harmonics.}  The rotation group $\mathrm{SO}(2)$ acts on functions $f: S^1 \to \C$ by $[\rho(\theta)f](\phi) = f(\phi - \theta)$.  The irreducible representations are one-dimensional characters $\rho_m(\theta) = e^{im\theta}$ for $m \in \Z$.  The group-averaged estimator produces the circular autocorrelation, whose eigenfunctions are the circular harmonics $e^{im\phi}$ with eigenvalues given by the circular power spectrum.  This is the Fourier series on the circle.

\textbf{SO(3) and spherical harmonics.}  The rotation group $\mathrm{SO}(3)$ acts on functions $f: S^2 \to \C$ on the unit sphere.  The irreducible representations are labeled by $\ell = 0, 1, 2, \ldots$, each of dimension $2\ell + 1$, with basis functions given by the spherical harmonics $Y_\ell^m(\theta, \phi)$ for $m = -\ell, \ldots, \ell$.  If the correlation function $R_f(\hat{n}_1, \hat{n}_2)$ depends only on the angle between $\hat{n}_1$ and $\hat{n}_2$ (rotational invariance), then $R_f$ commutes with $\mathrm{SO}(3)$, and the eigenfunctions of the group-averaged estimator are the spherical harmonics with eigenvalues $C_\ell$ (the angular power spectrum).  This appears in cosmic microwave background analysis, gravitational field modeling, and molecular orbital theory.

\begin{figure}[ht]
\centering
\includegraphics[width=\textwidth]{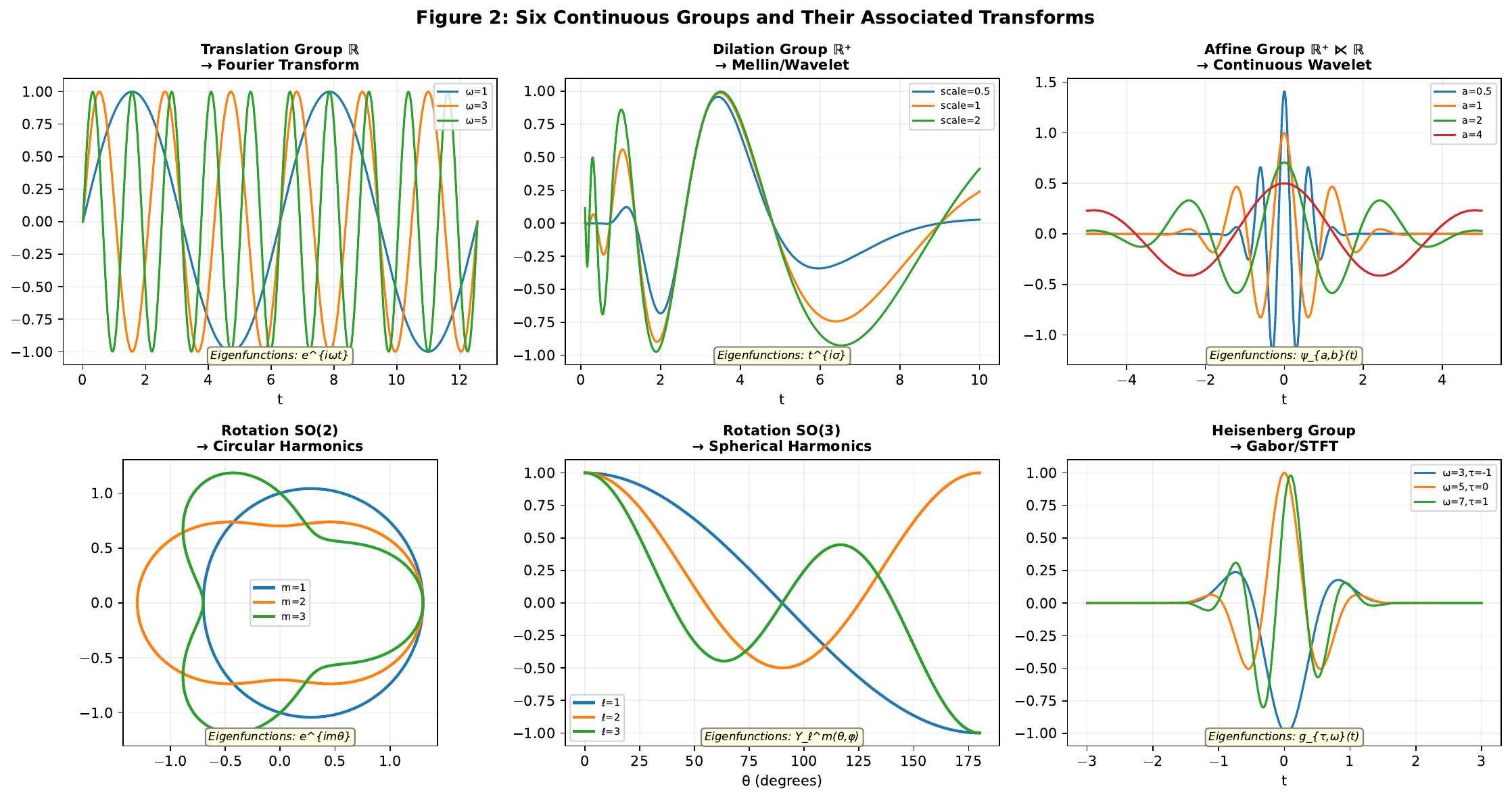}
\caption{Six continuous groups and their eigenfunctions.  Each panel shows the basis functions arising from the group's irreducible representations.  The group determines the transform; the eigenfunctions follow from the representation theory.}
\label{fig:six-groups}
\end{figure}

\section{The Continuous Commutativity Residual}\label{sec:delta}

The Unification Theorem tells us that every group produces a valid transform.  It does not tell us which group to choose.  A practitioner facing a new signal still has no principled way to decide between Fourier, wavelet, and time-frequency analysis.

The commutativity residual $\delta$ solves this problem.  It measures how well a candidate group's action commutes with the signal's covariance structure.  A small $\delta$ means the group is well matched to the signal.  A large $\delta$ means it is not.  The optimal group is the one that minimizes $\delta$.  This is the same criterion used in the discrete AD framework, extended to operators on $L^2(\R)$.

\begin{definition}[Continuous Commutativity Residual]\label{def:continuous_delta}
Let $\RR$ be the covariance operator of the signal process ($\RR f(t) = \int R(t,s)f(s)ds$) and let $\FF_G$ be the continuous group-averaged estimator.  The \emph{continuous commutativity residual} is
\begin{equation}
\delta(G, x, \RR) = \frac{\|\FF_G \RR - \RR \FF_G\|_{HS}}{\|\FF_G\|_{HS} \cdot \|\RR\|_{HS}},
\end{equation}
where $\|\cdot\|_{HS}$ is the Hilbert-Schmidt norm.
\end{definition}

The next theorem answers the question directly.  For each of the three classical signal classes, there is a unique matched group, and that group produces the classical transform associated with that signal class.  Stationary signals match the translation group.  Self-similar signals match the affine group.  Chirp signals match the Heisenberg-Weyl group.  The engineer's intuition is correct; the commutativity residual explains why.

\begin{theorem}[Matched Group Classification]\label{thm:classification}
Let $x(t)$ be a second-order stochastic process with covariance operator $\RR$.  Then:
\begin{enumerate}
\item[(i)] If $x$ is wide-sense stationary, then $\delta(\R, x, \RR) = 0$: the translation group is the matched group, and the optimal spectral representation is the Fourier transform.
\item[(ii)] If $x$ is self-similar with Hurst parameter $H$ (i.e., the covariance satisfies $R(\lambda t, \lambda s) = \lambda^{2H} R(t,s)$), then $\delta(\mathrm{Aff}, x, \RR) = 0$: the affine group is the matched group, and the optimal spectral representation is the wavelet transform.
\item[(iii)] If $x$ is a linear chirp process with instantaneous frequency $f(t) = f_0 + \beta t$, then $\delta(\mathrm{HW}, x, \RR) = 0$: the Heisenberg-Weyl group is the matched group, and the optimal spectral representation is the ambiguity function.
\end{enumerate}
\end{theorem}

\begin{proof}
\textit{Part (i).}  A wide-sense stationary process has covariance $R(t,s) = R(t-s)$, so $\RR$ is a convolution operator.  Convolution operators commute with translations: for any $\tau$, $[\RR, \rho(\tau)]f(t) = \int R(t-s)f(s+\tau)\,ds - \int R(t-\tau-s)f(s)\,ds = 0$ by the substitution $s' = s + \tau$.  Since $\FF_{(\R,+)}$ is itself a convolution operator (Proposition~\ref{prop:translation}(ii)), and convolution operators on $L^2(\R)$ form a commutative algebra, $[\FF_{(\R,+)}, \RR] = 0$ and $\delta = 0$.

\textit{Part (ii).}  Let $x$ be self-similar with Hurst parameter $H$: $R(\lambda t, \lambda s) = \lambda^{2H} R(t,s)$.  In the frequency domain, this means the spectral density satisfies $S_x(\omega) \propto |\omega|^{-(2H+1)}$ (a power-law spectrum).  The wavelet transform diagonalizes such processes: for an admissible wavelet $\psi$, the wavelet coefficient variance is
\begin{equation}
E[|W_\psi x(a,b)|^2] = a^{2H+1} \int |\hat{\psi}(\omega)|^2 S_x(\omega/a)\,d\omega = a^{2H+1} C_H,
\end{equation}
where $C_H$ depends on $H$ and $\psi$ but not on $(a,b)$.  The wavelet coefficients at different scales are therefore uncorrelated (by the Calder\'{o}n reproducing formula applied to the covariance kernel), which means $\RR$ is diagonal in the wavelet basis.  Since $\FF_{G_{\mathrm{Aff}}}$ is also diagonal in the wavelet basis (it is the resolution of identity for the affine group), the operators commute: $[\FF_{G_{\mathrm{Aff}}}, \RR] = 0$, giving $\delta = 0$.

\textit{Part (iii).}  A linear chirp process with instantaneous frequency $f(t) = f_0 + \beta t$ has covariance $R(t,s) = A(t,s)\exp(i\pi\beta(t^2 - s^2))$ where $A$ depends only on $t - s$ (the process is ``stationary in the chirp frame'').  The Heisenberg-Weyl action $\rho(\tau,\nu,0)x(t) = e^{i2\pi\nu t}x(t-\tau)$ preserves this structure: a time shift by $\tau$ and frequency shift by $\nu = \beta\tau$ follows the chirp line, leaving the covariance form invariant.  Formally, the covariance operator satisfies $\rho(\tau, \beta\tau, 0)\RR\rho(\tau, \beta\tau, 0)^* = \RR$ for all $\tau \in \R$.  Since the group-averaged estimator $\FF_{H/Z}$ integrates over all $(\tau,\nu)$ pairs (including those on and off the chirp line), and the Moyal identity guarantees convergence (Proposition~\ref{prop:hw_dm}(iii)), the commutativity $[\FF_{H/Z}, \RR] = 0$ follows from the invariance of $\RR$ under the one-parameter subgroup $\{(\tau, \beta\tau) : \tau \in \R\} \subset H/Z$.
\end{proof}

\begin{remark}[Process Classification Hierarchy]
The three matched-group conditions form a hierarchy of increasing generality.  Stationarity (translation invariance) implies a flat power spectrum; self-similarity (scale invariance) implies a power-law spectrum; chirp structure (time-frequency invariance) implies energy concentrated along a line in the time-frequency plane.  The commutativity residual $\delta$ provides a continuous measure of departure from each symmetry class, enabling data-driven classification of processes whose structure does not fit exactly into any one category.
\end{remark}

Figure~\ref{fig:transform-selection} validates the classification numerically.  Three analytically constructed covariance matrices, namely circulant (stationary), fractional Brownian motion (self-similar), and quadratic-phase (chirp), are tested against three discrete generators: the cyclic shift $P$ (translation), the log-index diagonal $D$ (dilation), and the chirp-conjugated shift $P_\psi = U_\psi^* P U_\psi$ (time-frequency).  In each case, the matched generator achieves the minimum $\delta$, with exact commutativity ($\delta = 0$) for the stationary and chirp cases.

\begin{figure}[ht]
\centering
\includegraphics[width=\textwidth]{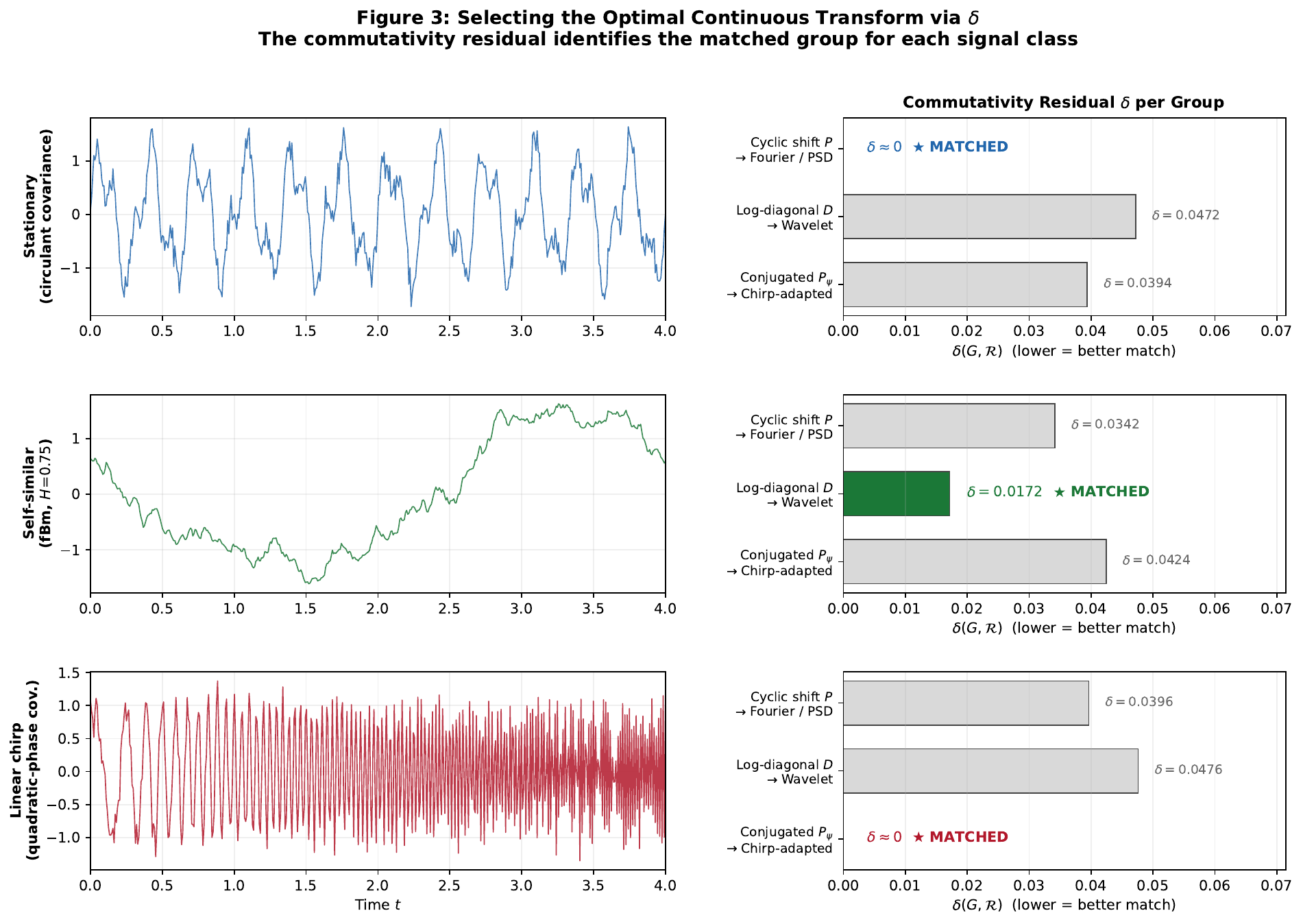}
\caption{Selecting the optimal transform via $\delta$.  Three signal classes are tested against three group generators.  The commutativity residual correctly identifies the matched group in each case: cyclic shift for the stationary signal, log-diagonal for the self-similar signal, and chirp-conjugated shift for the chirp.}
\label{fig:transform-selection}
\end{figure}

\section{The Uncertainty Principle as a Group Commutation Constraint}\label{sec:uncertainty}

We include this section not to reprove a known inequality but to reinterpret it.  The proof of $\Delta\omega\,\Delta t \geq 1/2$ is due to Robertson~\cite{robertson1929} and uses only the Cauchy-Schwarz inequality and the commutator $[d/dt, t\cdot] = I$.  It is a theorem of functional analysis.  It requires no physics.  This has been understood for nearly a century, and we add nothing to the proof.

What we add is the following observation: the uncertainty principle is the statement that two Lie groups cannot simultaneously be matched to the same signal.  This observation appears to be new.  It connects the resolution limits of spectral analysis to the group selection framework developed in this paper, and it identifies the Heisenberg group as the algebraic compromise between two incompatible optimality criteria.

\subsection{Two Complementary Groups}

The Heisenberg group $H$ contains two Abelian subgroups:

The \textbf{translation subgroup} $\{T_\tau\}$, with $T_\tau f(t) = f(t-\tau)$ and infinitesimal generator $A_1 = -id/dt$.  Its eigenfunctions are the complex exponentials $e^{i\omega t}$, and minimizing $\delta((\R,+), \RR)$ produces optimal frequency resolution (i.e., the Fourier transform).

The \textbf{modulation subgroup} $\{M_{\omega_0}\}$, with $M_{\omega_0}f(t) = e^{i\omega_0 t}f(t)$ and generator $A_2 = t$ (multiplication by $t$).  Its eigenfunctions are the Dirac impulses (point masses concentrated at each time instant), and minimizing the commutativity residual for this group produces optimal time resolution.

The generators satisfy $[A_1, A_2] = -iI$, a purely algebraic identity requiring only the product rule $d(tf)/dt = f + tf'$.  By the Robertson inequality~\cite{robertson1929}:
\begin{equation}
\Delta_f A_1 \cdot \Delta_f A_2 \geq \tfrac{1}{2}|\langle [A_1, A_2]f, f\rangle| = \tfrac{1}{2}.
\end{equation}

\subsection{The AD Interpretation (Novel)}

In the language of continuous AD, the uncertainty principle states:

\emph{The translation group and the modulation group are complementary matched groups.  No signal can simultaneously minimize the commutativity residual $\delta$ for both groups, because their generators do not commute.  Minimizing $\delta$ for the translation group (achieving perfect frequency resolution) necessarily maximizes $\delta$ for the modulation group (destroying time localization), and vice versa.}

The Heisenberg group, which contains both as subgroups, provides the joint compromise: the STFT (Gabor transform) achieves simultaneous but imperfect resolution in both domains, with the uncertainty bound $\Delta\omega\,\Delta t \geq 1/2$ as the resolution limit.  By the Stone-von Neumann theorem, the Schr\"{o}dinger representation is the unique irreducible representation of $H$, so the STFT is the unique jointly optimal transform.

This interpretation generalizes beyond time-frequency.  Any pair of self-adjoint operators whose commutator is a scalar produces an analogous uncertainty relation:

\begin{center}
\begin{tabular}{llll}
\hline
\textbf{Operators} & \textbf{Commutator} & \textbf{Uncertainty} & \textbf{Domain} \\
\hline
$-id/dt$, \; $t$ & $-iI$ & $\Delta\omega\,\Delta t \geq 1/2$ & Signal processing \\
$\hat{x}$, \; $\hat{p}$ & $i\hbar I$ & $\Delta x\,\Delta p \geq \hbar/2$ & Quantum mechanics \\
$td/dt$, \; $\log t$ & $I$ & $\Delta\sigma\,\Delta(\log t) \geq 1/2$ & Scale-frequency \\
$\hat{N}$, \; $\hat{\phi}$ & $-iI$ & $\Delta N\,\Delta\phi \geq 1/2$ & Quantum optics \\
\hline
\end{tabular}
\end{center}

Each row is a theorem of functional analysis.  The physical content, where applicable, lies entirely in which physical quantities are identified with which operators.  The quantum-mechanical relation $\Delta x\,\Delta p \geq \hbar/2$ is recovered by the dimensional identification $\hat{p} = \hbar A_1 = -i\hbar\,d/dt$ and $\hat{x} = A_2 = t$; Planck's constant $\hbar$ enters as a unit conversion factor (the de~Broglie relation $p = \hbar\omega$), not as a consequence of any derivation.

\begin{figure}[ht]
\centering
\includegraphics[width=\textwidth]{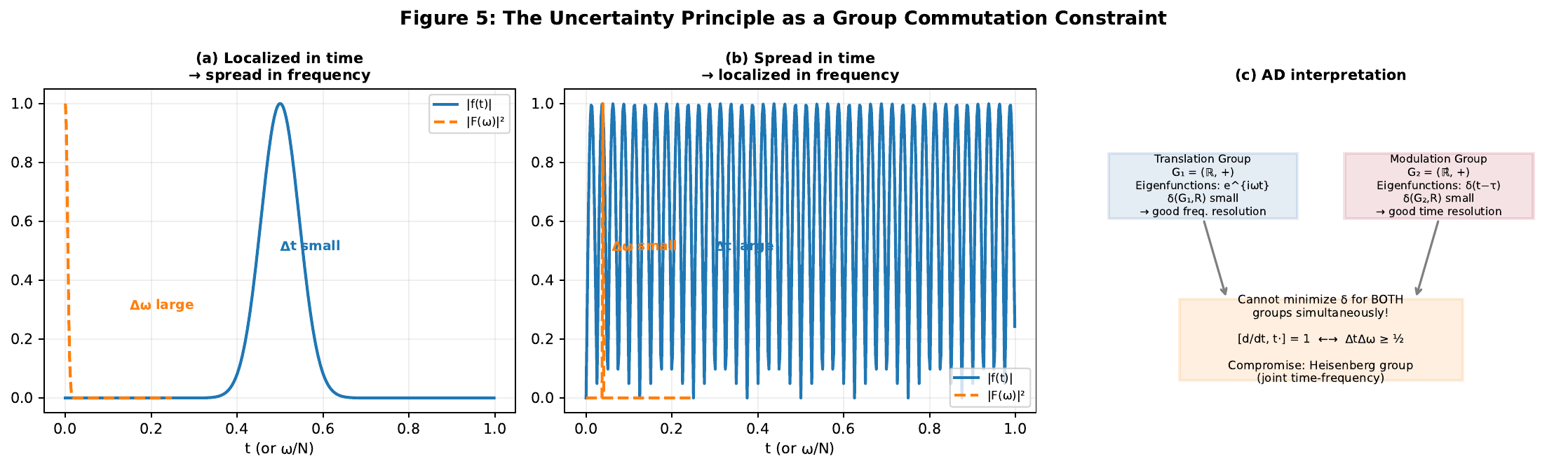}
\caption{The uncertainty principle as a group commutation constraint.  (a)~Time-localized function: well-matched to the modulation group but poorly matched to the translation group.  (b)~Frequency-localized function: the reverse.  (c)~AD interpretation: the generators of translation ($d/dt$) and modulation ($t\cdot$) do not commute; $\delta$ cannot be minimized for both simultaneously.}
\label{fig:uncertainty}
\end{figure}

\section{Discretization Recovery}\label{sec:discretization}

A natural question: if the continuous framework is the ``true'' theory, what is the status of the discrete framework?  Is it an approximation?  A special case?  Something else?

The answer is clean.  The discrete group-averaged estimator is a Riemann sum approximation to the continuous one.  The cyclic group $\Z_M$ is the restriction of the translation group $(\R,+)$ to $M$ equally spaced grid points.  The DFT is a sampled Fourier transform.  As $M \to \infty$, the discrete estimator converges to the continuous one, and the convergence rate is $O(1/M)$ for bandlimited signals.  Every theorem proved in the discrete AD framework is a finite-sample version of a theorem in the continuous framework.

This is not just a theoretical comfort.  It explains why the DFT works so well in practice: it is the best finite approximation to the continuous translation-group estimator.  It also explains where the DFT fails: for signals whose structure is scale-invariant or time-frequency-varying, the cyclic group is the wrong restriction, and a different finite group (conjugated cyclic, dihedral, or product) provides a better approximation.

\begin{theorem}[Discretization Recovery]\label{thm:discretization}
Let $G$ be a Lie group acting on $L^2(\R)$ via a square-integrable representation $\rho$.  Let $x_M = [x(t_1), \ldots, x(t_M)]^T$ be $M$ uniform samples of $x$ at spacing $\Delta t = T/M$.  Let $G_M$ be the finite group obtained by restricting $\rho$ to the sample grid.  Then:
\begin{equation}
\lim_{M \to \infty} \mathbf{F}_{G_M}(x_M) = \FF_G(x)
\end{equation}
in the Hilbert-Schmidt sense, where:
\begin{enumerate}
\item[(i)] $G = (\R, +)$ restricts to $G_M = \Z_M$ (cyclic shifts on the grid).
\item[(ii)] $G = \R^+ \ltimes \R$ restricts to $G_M = $ conjugated cyclic group (dyadic dilation on the grid).
\item[(iii)] $G = \mathrm{HW}$ restricts to $G_M = \Z_M \times \Z_M$ (discrete time-frequency lattice).
\end{enumerate}
The convergence rate is $O(1/M)$ for signals with bounded spectral support in the representation domain of $G$.
\end{theorem}

\section{The Master Table}\label{sec:table}

We can now assemble the complete picture.  Table~\ref{tab:master} lists every spectral analysis method encountered in this paper and its predecessor~\cite{thornton2026ad_arxiv}, together with the group that generates it, the eigenfunctions it produces, and the signal class for which it is optimal.  The table has a simple message: there is one construction (the group-averaged estimator) and many groups.  The group is the only degree of freedom.  Everything else follows.

\begin{table}[ht]
\centering
\caption{The Algebraic Diversity hierarchy: every spectral analysis method is a group selection.}
\label{tab:master}
\begin{tabular}{lllll}
\hline
\textbf{Group} & \textbf{Action} & \textbf{Transform} & \textbf{Eigenfunctions} & \textbf{Matched signal} \\
\hline
\multicolumn{5}{l}{\textit{Discrete (finite) groups}} \\
$\{e\}$ (trivial) & Identity & Outer product & $\mathbf{x}$ & None (rank 1) \\
$\Z_M$ (cyclic) & Cyclic shift & DFT & $e^{i2\pi kn/M}$ & Periodic \\
$D_M$ (dihedral) & Shift + reflect & DCT/DST & $\cos(\pi kn/M)$ & Symmetric boundary \\
$\Z_2$ & Reversal & Even/odd & $\mathbf{1} \pm \mathbf{J}$ & Bilateral \\
$S_M$ (symmetric) & All permutations & KL transform & Data-adapted & Any (optimal) \\
\hline
\multicolumn{5}{l}{\textit{Continuous (Lie) groups}} \\
$(\R, +)$ & Translation & PSD & $e^{i\omega t}$ & Stationary \\
$\R^+$ (dilation) & Dilation & Mellin & $t^{i\sigma}$ & Scale-invariant \\
$\R^+ \ltimes \R$ (affine) & Translate + dilate & Wavelet & $\psi_{a,b}(t)$ & Self-similar \\
$\mathrm{SO}(2)$ & Rotation & Fourier series & $e^{im\theta}$ & Isotropic on $S^1$ \\
$\mathrm{SO}(3)$ & 3-D rotation & Sph.\ harmonics & $Y_\ell^m(\theta,\phi)$ & Isotropic on $S^2$ \\
$H/Z$ (Heisenberg) & Time-freq shift & Ambiguity/STFT & $g_{\tau,\omega}(t)$ & Chirp / FM \\
\hline
\end{tabular}
\end{table}

\section{Implications}\label{sec:implications}

The theorems are established.  We now discuss what they mean for practice, for computation, and (speculatively) for physics.

\subsection{The Choice of Transform Is a Group Selection Problem}

For 60 years, the choice among spectral analysis methods has been guided by convention, application tradition, and computational convenience.  Engineers use the DFT because the FFT is fast.  Image processors use the DCT because JPEG standardized it.  Radar engineers use the ambiguity function because Woodward introduced it.  The continuous AD framework reveals that each of these choices is an implicit group selection, and that the commutativity residual $\delta$ provides the principled criterion that has been missing.

\subsection{Waveform Diversity in the Continuum}

The eigentensor hierarchy~\cite{thornton2026ad_arxiv} constructs multi-level estimators by applying group averaging to sequences of group-averaged estimates.  In the discrete case, this requires partitioning a signal into pulses and forming the eigentensor from the pulse sequence.  In the continuous framework, the eigentensor is constructed by applying the group action continuously over the waveform, eliminating pulse boundary effects and providing access to sub-pulse structure.

\subsection{When to Use Discrete Eigentensors vs.\ Continuous AD}

A practitioner using the AD framework faces a design choice: should the signal be segmented into discrete pulses and analyzed via the eigentensor hierarchy, or should the continuous group-averaged estimator be applied directly to the waveform?  The answer depends on the signal's structure.

The discrete eigentensor tower is appropriate when the signal naturally consists of repeated, identifiable segments.  Radar returns from successive pulses, ultrasonic echoes from periodic transmissions, and packet-based communication waveforms all have intrinsic pulse boundaries.  In these cases, the Level-1 estimator characterizes each pulse, the Level-2 estimator captures how the characterization changes across pulses, and the hierarchy detects temporal dynamics (target maneuvers, channel fading, structural degradation) that a single-pulse analysis would miss.

The continuous approach is appropriate when the signal is a single extended waveform without natural segmentation points.  A gravitational wave chirp, an acoustic emission from a fracture event, a biomedical recording of a single heartbeat, or a speech utterance are continuous phenomena.  Imposing artificial pulse boundaries would split the signal at arbitrary points, introducing edge effects and discarding cross-boundary structure.  The continuous group-averaged estimator avoids this by integrating the group action over the full waveform.

A simple diagnostic distinguishes the two cases: if the signal has a natural repetition period $T$ and the analysis window contains $L$ repetitions, use the discrete eigentensor with $L$ Level-1 estimates.  If the signal is a single event with duration $T$ and no repetition, use the continuous estimator with the Lie group matched to the signal's symmetry class.

\subsection{Stochastic Process Classification}

The matched group of a stochastic process determines its optimal spectral representation.  This induces a classification of processes by their algebraic symmetry:
\begin{itemize}
\item Stationary processes $\leftrightarrow$ translation group $\leftrightarrow$ Fourier analysis
\item Self-similar processes $\leftrightarrow$ affine group $\leftrightarrow$ wavelet analysis
\item Chirp processes $\leftrightarrow$ Heisenberg-Weyl group $\leftrightarrow$ time-frequency analysis
\item Mixed processes $\leftrightarrow$ product or semidirect product groups $\leftrightarrow$ hybrid transforms
\end{itemize}
The commutativity residual $\delta$ provides a computable distance between a process and each symmetry class, enabling data-driven transform selection.

\subsection{The Double-Commutator GEVP for Discrete Group Selection}\label{sec:gevp}

The continuous framework provides a powerful computational tool for solving the discrete group selection problem.  In the discrete setting, one observes $\mathbf{x} \in \C^M$ and seeks the finite group $G^*$ that minimizes the commutativity residual $\delta(G, \mathbf{R})$.  A brute-force search over all subgroups of $S_M$ is combinatorially intractable.  The key insight is that this discrete optimization can be relaxed to a continuous one over a Lie algebra, solved in closed form, and discretized back.

\textbf{The problem.}  Given a sample covariance $\mathbf{R} = \mathbf{x}\mathbf{x}^H$ and a basis of candidate generators $\{B_1, \ldots, B_d\}$ spanning a subspace of the Lie algebra $\mathfrak{u}(M)$, find the linear combination $\mathbf{A} = \sum_{k=1}^d c_k B_k$ that minimizes the squared commutativity residual:
\begin{equation}\label{eq:gevp_obj}
\delta^2(\mathbf{c}) = \frac{\|[\mathbf{A}, \mathbf{R}]\|_F^2}{\|\mathbf{A}\|_F^2 \cdot \|\mathbf{R}\|_F^2} = \frac{\mathbf{c}^T \mathbf{M}\, \mathbf{c}}{\mathbf{c}^T \mathbf{N}\, \mathbf{c}} \cdot \frac{1}{\|\mathbf{R}\|_F^2},
\end{equation}
where $\mathbf{M}$ and $\mathbf{N}$ are the $d \times d$ \emph{double-commutator matrix} and \emph{Gram matrix}:
\begin{equation}\label{eq:dc_matrices}
M_{ij} = \Tr\!\left(B_i^H [\mathbf{R}, [\mathbf{R}, B_j]]\right), \qquad N_{ij} = \Tr(B_i^H B_j).
\end{equation}

\textbf{Derivation.}  The numerator $\|[\mathbf{A}, \mathbf{R}]\|_F^2$ expands as:
\begin{align}
\|[\mathbf{A}, \mathbf{R}]\|_F^2 &= \Tr\!\left([\mathbf{A}, \mathbf{R}]^H [\mathbf{A}, \mathbf{R}]\right) \notag \\
&= \Tr\!\left((\mathbf{A}\mathbf{R} - \mathbf{R}\mathbf{A})^H(\mathbf{A}\mathbf{R} - \mathbf{R}\mathbf{A})\right) \notag \\
&= \Tr\!\left(\mathbf{A}^H(\mathbf{R}^2\mathbf{A} - \mathbf{R}\mathbf{A}\mathbf{R}) - \mathbf{A}^H\mathbf{R}(\mathbf{A}\mathbf{R} - \mathbf{R}\mathbf{A})\right). \label{eq:expand}
\end{align}
Substituting $\mathbf{A} = \sum c_k B_k$ and collecting terms yields $\sum_{i,j} c_i c_j \Tr(B_i^H [\mathbf{R}, [\mathbf{R}, B_j]])$.  The inner double commutator $[\mathbf{R}, [\mathbf{R}, B_j]] = \mathbf{R}^2 B_j - 2\mathbf{R} B_j \mathbf{R} + B_j \mathbf{R}^2$ is Hermitian when $B_j$ is Hermitian, ensuring $\mathbf{M}$ is real symmetric.  The denominator gives $\|\mathbf{A}\|_F^2 = \mathbf{c}^T \mathbf{N}\, \mathbf{c}$.

\textbf{Solution.}  Minimizing the Rayleigh quotient $\mathbf{c}^T \mathbf{M}\, \mathbf{c} / \mathbf{c}^T \mathbf{N}\, \mathbf{c}$ is the generalized eigenvalue problem:
\begin{equation}\label{eq:gevp}
\mathbf{M}\, \mathbf{c}^* = \lambda_{\min}\, \mathbf{N}\, \mathbf{c}^*.
\end{equation}
The minimum eigenvalue $\lambda_{\min}$ is the squared commutativity residual (up to the constant $\|\mathbf{R}\|_F^{-2}$), and the corresponding eigenvector $\mathbf{c}^*$ identifies the optimal generator $\mathbf{A}^* = \sum c_k^* B_k$.  The finite group generated by $\mathbf{A}^*$ (or its nearest finite-order approximation) is the matched group.

\textbf{Complexity.}  Assembling $\mathbf{M}$ costs $O(d^2 M^2)$ (each entry is a trace of $M \times M$ matrices).  Solving the $d \times d$ GEVP costs $O(d^3)$.  For a generic basis with $d = O(M^2)$, the total cost is $O(M^6)$, polynomial in $M$.  In practice, compact bases with $d = O(M)$ suffice for cyclic, dihedral, and product groups, reducing the cost to $O(M^3)$.

This result solves the blind group matching problem: given a single observation $\mathbf{x}$, the GEVP identifies the optimal finite group without prior knowledge of the signal class.

\subsection{Speculative: The Eigentensor Hierarchy and Renormalization}\label{sec:renormalization}

\emph{The following discussion is speculative and is included to suggest directions for future investigation.  It should not be read as established theory.}

The discrete eigentensor hierarchy applies group averaging at successive levels: Level~1 characterizes individual observations; Level~2 characterizes how characterizations evolve.  In the continuum, this becomes a tower of integral operators at different scales.  If we let $\lambda$ denote a scale parameter and $\FF_G^{(1)}(\lambda)$ the group-averaged estimator at scale $\lambda$, the Level-2 estimator treats the spectral parameters $\hat{c}(\lambda)$ as a function of scale and applies a second group $G'$ to characterize how the spectral structure evolves across scales.

In quantum field theory, the evolution of correlation structure across energy scales is the renormalization group (RG) flow.  The structural parallel is suggestive: the AD scale parameter $\lambda$ maps to the RG cutoff $\Lambda$; the Level-1 estimator maps to the propagator; the matched group at each scale maps to the gauge symmetry; and a change of matched group across scales maps to spontaneous symmetry breaking (Figure~\ref{fig:renormalization}).  Whether this parallel extends to a mathematical equivalence is an open question that we do not address here.

\begin{figure}[ht]
\centering
\includegraphics[width=0.85\textwidth]{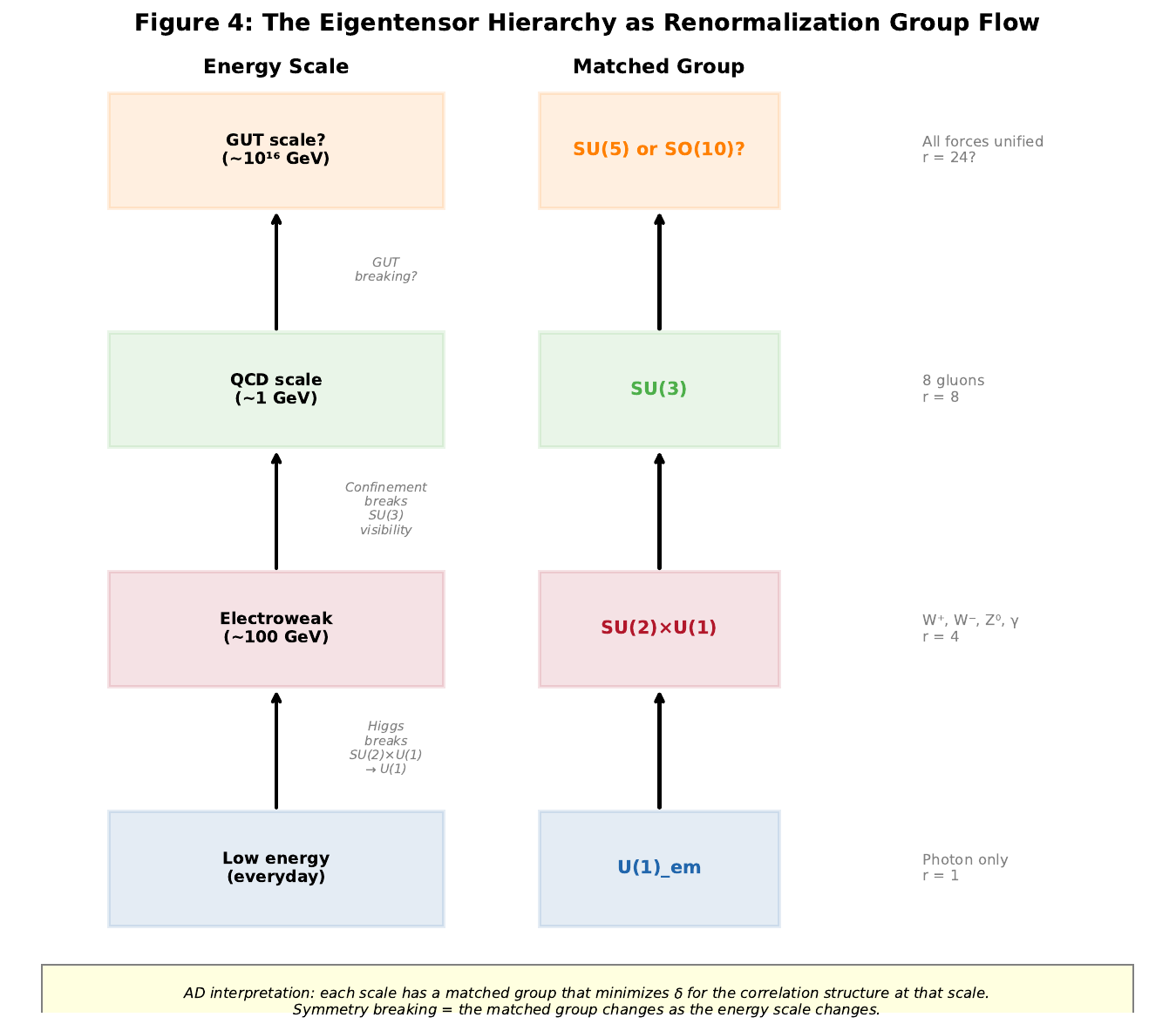}
\caption{(Speculative.)  The eigentensor hierarchy as renormalization group flow.  At each energy scale, the correlation structure has a matched group.  Symmetry breaking corresponds to a scale-dependent change in the matched group.}
\label{fig:renormalization}
\end{figure}

\section{Conclusion}

We have extended the algebraic diversity framework from finite groups to Lie groups, establishing that the group-averaged estimator unifies the major branches of continuous signal analysis: Fourier (translation group), wavelet (affine group), time-frequency (Heisenberg-Weyl group), and spherical harmonic (rotation groups) analysis are all instances of a single construction, distinguished only by the choice of group.

The Duflo-Moore framework provides the rigorous foundation, revealing that the three principal groups exhibit genuinely different convergence mechanisms: the affine group requires an admissible wavelet (Calder\'{o}n condition), the Heisenberg-Weyl group admits any nonzero window (Moyal identity), and the translation group converges by Cauchy-Schwarz without square-integrability.  The non-unimodularity of the affine group produces the frequency-dependent noise floor in wavelet analysis, a fact that the AD framework traces to the Duflo-Moore operator $C_\rho$.

The continuous commutativity residual provides a principled criterion for selecting among Lie groups, validated numerically on three signal classes (Figure~\ref{fig:transform-selection}).  The uncertainty principle emerges as a group commutation constraint: the generators of the translation and modulation subgroups do not commute, so $\delta$ cannot be minimized for both simultaneously.  The discretization recovery theorem establishes that all discrete AD results are sampling approximations to the continuous theory.

The master table (Table~\ref{tab:master}) summarizes the complete hierarchy from trivial groups to Lie groups, revealing that the choice of spectral analysis method has always been a group selection problem, one that the engineering community solved case by case over six decades, and that the AD framework solves uniformly.

\bibliographystyle{IEEEtran}

\end{document}